\documentclass{article}
\usepackage{arxiv}

\usepackage{amsthm}

\usepackage{xcolor}  

\usepackage{amssymb}

\usepackage{amsmath}
\usepackage{mathdots}

\usepackage{mathtools}

\usepackage{tensor}

\usepackage{urwchancal}
\DeclareFontFamily{OT1}{pzc}{}
\DeclareFontShape{OT1}{pzc}{m}{it}{<-> s * [1.10] pzcmi7t}{}
\DeclareMathAlphabet{\mathpzc}{OT1}{pzc}{m}{it}

\usepackage{graphicx}
\usepackage{ifpdf}
\ifpdf
\usepackage{epstopdf}
\PrependGraphicsExtensions{.svg}
\fi

\newtheorem{theorem}{Theorem}[section]
\newtheorem{lemma}[theorem]{Lemma}
\newtheorem{corollary}[theorem]{Corollary}

\newtheorem{remark}[theorem]{Remark}

\providecommand{\R}{\mathbb{R}}

\providecommand{\SO}{\mathbf{SO}}

\providecommand{\GL}{\mathbf{GL}}
\providecommand{\SE}{\mathbf{SE}}

\providecommand{\SIM}{\mathbf{SIM}}

\providecommand{\Aut}{\mathbf{Aut}} 

\providecommand{\gothsim}{\mathfrak{sim}}

\providecommand{\gothX}{\mathfrak{X}} 

\providecommand{\so}{\mathfrak{so}}
\providecommand{\se}{\mathfrak{se}}

\providecommand{\aut}{\mathfrak{aut}}

\providecommand{\calU}{\mathcal{U}}

\providecommand{\tT}{\mathrm{T}} 

\providecommand{\eb}{\mathbf{e}} 

\DeclareMathOperator{\tr}{tr}
\providecommand{\trace}[1]{\tr\left(#1\right)}

\DeclareMathOperator{\diag}{diag}

\providecommand{\Lyap}{\mathcal{L}} 

\providecommand{\td}{\mathrm{d}}

\providecommand{\ddt}{\frac{\td}{\td t}}

\usepackage{accents}
\makeatletter
\providecommand{\scirc}{%
    \hbox{\fontfamily{\rmdefault}\fontsize{0.4\dimexpr(\f@size pt)}{0}\selectfont{\raisebox{-0.52ex}[0ex][-0.52ex]{$\circ$}}}}

\makeatother

\mathchardef\mhyphen="2D

\providecommand{\gl}{\mathfrak{gl}}

\providecommand{\trans}{ V}
\providecommand{\scale}{A}
\providecommand{\tvel}{ W }
\providecommand{\svel}{S}

\usepackage{graphicx}      
\usepackage{natbib}        
\usepackage{hyperref}

\usepackage{verbatim}
\usepackage{ifthen}

\newcommand{\papercitation}{
To appear in IFAC-PapersOnLine 2023
}

\newcommand{\publicationdetails}
{
\papercitation \\
\copyrightNoticeIFACAccepted{2023} 
}
\newcommand{\publicationversion}
{Author accepted version}

\begin{document}





\headertitle{Constructive Equivariant Observer Design for Inertial Navigation}
\title{Constructive Equivariant Observer Design for Inertial Navigation}

\author{
\href{https://orcid.org/0000-0003-4391-7014}{\includegraphics[scale=0.06]{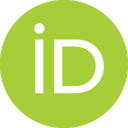}\hspace{1mm}
Pieter van Goor}
\\
    Systems Theory and Robotics Group \\
	Australian National University \\
    ACT, 2601, Australia \\
    \texttt{Pieter.vanGoor@anu.edu.au} \\
	\And	\href{https://orcid.org/0000-0002-7779-1264}{\includegraphics[scale=0.06]{orcid.png}\hspace{1mm}
    Tarek Hamel}
\\
    I3S (University C\^ote d'Azur, CNRS, Sophia Antipolis) \\
    and Insitut Universitaire de France \\
    \texttt{THamel@i3s.unice.fr} \\
	\And	\href{https://orcid.org/0000-0002-7803-2868}{\includegraphics[scale=0.06]{orcid.png}\hspace{1mm}
    Robert Mahony}
\\
    Systems Theory and Robotics Group \\
	Australian National University \\
    ACT, 2601, Australia \\
	\texttt{Robert.Mahony@anu.edu.au} \\
}

\maketitle

\vspace{1cm}

\begin{abstract} 
Inertial Navigation Systems (INS) are algorithms that fuse inertial measurements of angular velocity and specific acceleration with supplementary sensors including GNSS and magnetometers to estimate the position, velocity and attitude, or extended pose, of a vehicle.
The industry-standard extended Kalman filter (EKF) does not come with strong stability or robustness guarantees and can be subject to catastrophic failure.
This paper exploits a Lie group symmetry of the INS dynamics to propose the first nonlinear observer for INS with error dynamics that are almost-globally asymptotically and locally exponentially stable, independently of the chosen gains.
The observer is aided only by a GNSS measurement of position.
As expected, the convergence guarantee depends on persistence of excitation of the vehicle's specific acceleration in the inertial frame.
Simulation results demonstrate the observer's performance and its ability to converge from extreme errors in the initial state estimates.
\end{abstract}


\section{Introduction}


Inertial Navigation Systems (INS) are systems that estimate the \emph{extended pose}; that is, the pose, comprising position and attitude along with the linear velocity, of a rigid-body.
These particular variables are chosen since an inertial measurement unit (IMU) composed of a 3-axis gyroscope and a 3-axis accelerometer, measuring the angular velocity and specific acceleration, respectively, can theoretically be forward-integrated exactly to compute the evolution of the extended pose.
In practice, noise associated with MEMS\footnote{Micro Electrical Mechanical Systems} IMU measurements lead such methods to diverge quickly from the true extended pose \citep{2007_woodman_IntroductionInertialNavigation} and additional supporting sensors such as GNSS position, GNSS velocity, magnetometers, etc., must be fused into the state estimate.
Due to the importance of the application and the diversity of sensors there are a number of algorithms and approaches still under consideration in the literature, despite the fact that this problem has been under active research since the 1950s.

The archetypical INS problem is the fusion of IMU measurements with GNSS measurements, with the extended Kalman filter (EKF) providing the industry-standard solution \citep{1979_maybeck_StochasticModelsEstimation,2005_george_TightlyCoupledINS}.
A key disadvantage of EKF solutions, however, is that they do not provide guarantees of stability outside of a local trajectory-dependent domain. This has led to a number of nonlinear observers being proposed with more powerful stability and convergence guarantees.
In a particularly early work, \cite{2001_vik_NonlinearObserverGPS} developed a nonlinear observer for the GNSS-aided INS problem that additionally relied on the availability of a direct measurement of vehicle attitude in the solution.
Advances in techniques for attitude estimation from IMU and magnetometer measurements, including the complementary filter \citep{2008_mahony_NonlinearComplementaryFilters} and the invariant extended Kalman filter (IEKF) \citep{2007_bonnabel_LeftinvariantExtendedKalman}, inspired a number of new approaches to nonlinear designs for INS.
\cite{2011_barczyk_InvariantExtendedKalman} proposed an IEKF design for INS aided by GNSS position and magnetometer measurements using the symmetry group $\R^3\times \SO(3)\times \R^3\times \R^3$, and provided results demonstrating improved performance over a conventional EKF design.
\cite{2012_grip_ObserversInterconnectedNonlinear} developed a generic observer design for connecting a nonlinear and a linear observer such that the resulting combined observer is exponentially stable under appropriate assumptions on the chosen gains and initial system conditions.
They then applied this generic design to provide a solution for INS aided by GNSS position and velocity and magnetometer measurements.
A semi-globally stable nonlinear observer for INS with magnetometer and GNSS position was presented by \cite{2013_grip_NonlinearObserverGNSSaided}, and took into account the curvature and rotation of the Earth.
This was extended in \cite{2017_hansen_NonlinearObserverDesign} to compensate for the time-delay typical in real-world GNSS measurements.
\cite{2017_barrau_InvariantExtendedKalman} developed a general theory of IEKF, and demonstrated an application INS using the novel extended Special Euclidean group $\SE_2(3)$.
They showed that the INS dynamics are `group-affine' under this symmetry, and that this leads to state-independent error dynamics and a trajectory-independent domain of local stability for the IEKF.
Recently, \cite{2019_berkane_PositionVelocityAttitude} proposed a semi-globally exponentially stable nonlinear observer for INS with magnetometer and generic position measurements; that is, measurements of position that may, for example, include only horizontal components or a single range component.
\cite{2021_berkane_NonlinearNavigationObserver} then extended their previous work to show that the system's observability is characterised entirely by the persistence of excitation of the position measurements, and provided real-world experimental validation of the proposed observer.
The nonlinear observer community has made many significant contributions to the INS problem, including designs featuring trajectory-independent local stability and semi-global stability; however, no existing observer is almost-globally asymptotically and locally exponentially stable.

This paper builds on the authors' recent work on constructive observer design for group-affine systems \citep{2021_vangoor_AutonomousErrorConstructive} and uses this methodology to develop a novel GNSS position-aided INS solution.
Specifically, the Lie group observer architecture proposed in \citep{2021_vangoor_AutonomousErrorConstructive} is applied to the INS problem, and correction terms are identified through a Lyapunov design.
To the authors understanding, this is the first work that provides almost-global (apart from a set of measure zero) asymptotic and local exponential, stability guarantees of the error dynamics \citep{2012_grip_ObserversInterconnectedNonlinear,2017_barrau_InvariantExtendedKalman,2019_berkane_PositionVelocityAttitude}.
In contrast to other recent nonlinear INS observers, the proposed observer does not rely on magnetometer measurements and requires only persistence of excitation of the vehicle's specific acceleration in the inertial frame.
A simulation experiment demonstrates the convergence of the observer from an extreme initial condition where the difference between the true and estimated attitude is $0.99\pi$~rad.

Although the proposed observer may not provide the same performance as a stochastic filter when the initial condition is close to the true system state, it is a useful addition to avionics systems to provide (almost)-global robustness guarantees.
In addition, for small RPAS vehicles, where stochastic filtering suffers badly from highly nonlinear noise characteristics, observers of this type are often the only technology that works.

\section{Preliminaries}

The special orthogonal group is the Lie group of 3D rotations,
\begin{align*}
    \SO(3) := \{
        R \in \R^{3 \times 3} \mid R^\top R = I_3, \; \det(R) = 1
    \}.
\end{align*}
For any vector $\Omega \in \R^3$, define
\begin{align*}
    \Omega^\times
    = \begin{pmatrix}
        0 & -\Omega_3 & \Omega_2 \\
        \Omega_3 & 0 & -\Omega_1 \\
        -\Omega_2 & \Omega_1 & 0
    \end{pmatrix}.
\end{align*}
Then $\Omega^\times v = \Omega \times v$ for any $v \in \R^3$ where $\times$ is the usual vector (cross) product.
The Lie algebra of $\SO(3)$ is defined
\begin{align*}
    \so(3) := \{
        \Omega^\times \in \R^{3 \times 3} \mid \Omega \in \R^3
    \}.
\end{align*}
For any two vectors $a,b \in \R^3$, one has the following identities:
\begin{align}
    a^\times b &= -b^\times a, &
    (a^\times)^\top &= -a^\times, \notag \\
    a^\times b^\times &= b a^\top - a^\top b I_3, &
    (a \times b)^\times &= b a^\top - a b^\top.
    \label{eq:so3_identities}
\end{align}

The extended special Euclidean group and its Lie algebra are defined
\begin{align*}
    \SE_2(3) &:= \left\{
    \begin{pmatrix}
        R &   \trans  \\
        0  & I_2 
    \end{pmatrix} \in \R^{5 \times 5}
    \;\middle\vert\;
     R \in \SO(3), \; \trans \in \R^{3 \times 2}
    \right\}, \\
    \se_2(3) &:= \left\{
    \begin{pmatrix}
        \Omega^\times & \tvel \\ 0_{2\times 3} & 0_{2\times 2}
    \end{pmatrix} \in \R^{5 \times 5}
    \;\middle\vert\;
    \Omega \in \R^3, \; \tvel \in \R^{3 \times 2}
    \right\}.
\end{align*}
Here $V = \begin{pmatrix} v & x \end{pmatrix} \in \R^{3 \times 2}$  is thought of as coordinates for an element of the tangent bundle $\tT \R^3$; that is, a velocity $v \in \tT_x \R^3$ at base point $x \in \R^3$.
An element of $\SE_2(3)$ may be denoted $X = (R,\trans)$ for convenience, where $R \in \SO(3)$ and $\trans \in \R^{3 \times 2}$.
Likewise, an element of $\se_2(3)$ may be denoted $\Delta = (\Omega_\Delta, \tvel_\Delta)$, where $\Omega_\Delta \in \R^3$ and $\tvel_\Delta \in \R^{3 \times 2}$.

For any $A,B \in \R^{n \times n}$, the matrix commutator is given by
\begin{align*}
    [A, B] := AB - BA.
\end{align*}

A signal $y(t) \in \R^3$ is said to be \emph{persistently exciting} if there exist $\mu, T > 0$ such that, for all $t \geq 0$
\begin{align}
    \int_t^{t+T} \vert b \times y(\tau) \vert^2 \td \tau \geq \mu,
    \label{eq:persistence_excitation}
\end{align}
for all directions $b \in \R^3$ with $\vert b \vert = 1$.

\subsection{Automorphisms of $\SE_2(3)$}

An \emph{automorphism} of $\SE_2(3)$ is a diffeomorphism $\sigma : \SE_2(3) \to \SE_2(3)$ such that $\sigma(I) = I$ and $\sigma(XY) = \sigma(X)\sigma(Y)$. 
The set of all such maps, denoted $\Aut(\SE_2(3))$, is a Lie-group \cite{}. 
The following development provides an explicit realisation of $\Aut(\SE_2(3))$.
 
Define the matrix Lie group $\SIM_2(3)$ and its Lie algebra $\gothsim_2(3)$ to be
\begin{align*}
    \SIM_2(3) &:=
    \left\{
        \begin{pmatrix}
            R & \trans \\ 0_{2\times 3} & \scale 
        \end{pmatrix}
        \,\middle|\,
        R \in \SO(3),
        \trans \in \R^{3\times 2},
        \scale \in \GL(2)
    \right\}, \\
    \gothsim_2(3) &:=
    \left\{
        \begin{pmatrix}
            \Omega^\times & \tvel \\ 0_{2\times 3} & \svel
        \end{pmatrix}
        \,\middle|\,
        \Omega \in \R^3,
        \tvel \in \R^{3\times 2},
        \svel \in \gl(2)
    \right\}.
\end{align*}
An element of $\SIM_2(3)$ may be denoted $Z = (R_Z,\trans_Z,\scale_Z)$ for convenience, where $R_Z \in \SO(3), \trans_Z \in \R^{3\times 2}, \scale_Z \in \GL(2)$.
Likewise, an element of $\gothsim_2(3)$ may be denoted $\Gamma = (\Omega_\Gamma, \tvel_\Gamma, \svel_\Gamma)$, where $\Omega_\Gamma \in \R^3, \tvel_\Gamma \in \R^{3 \times 2}, \svel_\Gamma \in \gl(2)$.

\begin{lemma}
Let $\sigma_Z: \SE_2(3) \to \SE_2(3)$ be defined by $\sigma_Z(X) = Z X Z^{-1}$, in the sense of matrix multiplication, where $Z = (R_Z, \trans_Z, \scale_Z) \in \SIM_2(3)$.
Then $\sigma_Z$ is an automorphism of $\SE_2(3)$; i.e. $\sigma_Z \in \Aut(\SE_2(3))$.
\end{lemma}

\begin{proof}
It is easy to see that $\sigma_Z(I_5) = I_5$ and that $\sigma_Z(XY) = \sigma_Z(X) \sigma_Z(Y)$.
It is also straightforward to identify the inverse $\sigma_Z^{-1}(X) := Z^{-1} X Z$.
What remains is to show that $\sigma_Z(X) \in \SE_2(3)$ for any $Z \in \SIM_2(3)$ and $X \in \SE_2(3)$.
Let $X = (R_X, \trans_X) \in \SE_2(3)$.
Direct computation yields
\begin{align*}
    \sigma_Z(X) &= Z X Z^{-1}, \\
    &= \begin{pmatrix}
        R_Z & \trans_Z \\ 0_{2 \times 3} & \scale_Z
    \end{pmatrix}
    \begin{pmatrix}
        R_X & \trans_X \\ 0_{2 \times 3} & I_2
    \end{pmatrix}
    \begin{pmatrix}
        R_Z & \trans_Z \\ 0_{2 \times 3} & \scale_Z
    \end{pmatrix}^{-1}, \\
    &= \begin{pmatrix}
        R_Z R_X R_Z^\top &\,& R_Z \trans_X \scale_Z^{-1} + (I_3 - R_Z R_X R_Z^\top) \trans_Z \scale_Z^{-1} \\ 0_{2 \times 3} &\,& I_2
    \end{pmatrix}.
\end{align*}
Since $R_Z R_X R_Z^\top \in \SO(3)$, it is clear that $\sigma_Z(X) \in \SE_2(3)$, as required.
Therefore, $\sigma_Z$ is an automorphism of $\SE_2(3)$.
\end{proof}

The Lie algebra of $\Aut(\SE_2(3))$ is denoted $\aut(\SE_2(3))$ and consists of the vector fields $f \in \gothX(\SE_2(3))$ such that $f(XY) = f(X) Y + X f(Y)$ and $f(I) = 0$.

\begin{corollary}
Let $\Gamma = (\Omega_\Gamma, \tvel_\Gamma, \svel_\Gamma) \in \gothsim_2(3)$.
Then the vector field $f_\Gamma \in \gothX(\SE_2(3))$ defined by $f_\Gamma(X) = [\Gamma, X] = \Gamma X - X \Gamma$ is an element of $\aut(\SE_2(3))$; that is,
\begin{align*}
    f_\Gamma(I) &= 0, &
    f_\Gamma(XY) &= f_\Gamma(X) Y + X f_\Gamma(Y).
\end{align*}
\end{corollary}

It is interesting to note that, for $\Gamma \in \gothsim_2(3)$ and $X \in \SE_2(3)$, one has that $\Gamma X - X \Gamma$ lies in the tangent space at $X$, although this is not the case for $X\Gamma$ or $\Gamma X$ individually.

\begin{remark}
To the best of the authors' knowledge, the Lie group $\SIM_2(3)$ and its role as a realisation of the automorphisms of $\SE_2(3)$, is novel in this paper.
The notation is adapted from the computer vision literature, where $\mathbf{SIM}(3)$ is well understood as the group of 3D translations, rotations and scalings.
The extension to $\mathbf{SIM}_m(n)$ and its role as automorphims for $\SE_m(n)$, is straightforward.
\end{remark}

\section{Problem Description}

Consider a vehicle equipped with an inertial measurement unit (IMU).
Let $R \in \SO(3)$,  $v \in \R^3$, and $p \in \R^3$ denote the vehicle's attitude, velocity, and position, respectively, all with respect to some given inertial frame \{0\}.
The angular velocity $\Omega \in \R^3$ and the specific acceleration $a \in \R^3$ are measured by the IMU.
The dynamical model considered is
\begin{align}
    \dot{R} &= R \Omega^\times, &
    \dot{v} &= R a + g, &
    \dot{p} &= v,
    \label{eq:system_dynamics}
\end{align}
where $g \in \R^3$ is the gravity vector in the inertial frame (typically $g \approx 9.81\eb_3$).
We assume a measurement of the vehicle's position in the inertial frame,
 \begin{align}
    h(R, v, p) &= p,
    \label{eq:system_measurement}
\end{align}
is available.
Such measurements are provided by a GNSS.
The problem is to design an observer to estimate $(R, v, p)$.

\subsection{Lie Group Observer Architecture}

Let $X = (R, (v \; p)) \in \SE_2(3)$.
Then the dynamics \eqref{eq:system_dynamics} may be written as
\begin{align*}
    \dot{X} = X U + G X - [D,X],
\end{align*}
where
\begin{align*}
    U &= (\Omega, (a \; 0)) \in \se_2(3), &
    G &= (0, (g \; 0)) \in \se_2(3),  \\
    D &= (0, 0, \svel_D) \in \gothsim_2(3), &
    \svel_D &= \begin{pmatrix} 0 & 1 \\ 0 & 0 \end{pmatrix} \in \R^{2\times 2}.
\end{align*}

Let the state estimate $\hat{X} = (\hat{R}, (\hat{v} \; \hat{p})) \in \SE_2(3)$.
We introduce an auxiliary state $Z = (R_Z, (v_Z \; p_Z), \scale_Z) \in \SIM_2(3)$.
Then we propose the following observer dynamics,
\begin{align}
    \dot{\hat{X}} &= \hat{X} U + G \hat{X} - [D,\hat{X}] + Z \Delta Z^{-1} \hat{X}, \notag \\
    \dot{Z} &= (G - D) Z + Z \Gamma, \label{eq:observer-architecture}
\end{align}
where $\Delta \in \se_2(3)$ and $\Gamma \in \gothsim_2(3)$ are correction terms that will be designed later.
Define the observer error
\begin{align}
    \bar{E} := \sigma_Z^{-1}(X\hat{X}^{-1}) = Z^{-1} X \hat{X}^{-1} Z.
\end{align}
Then $\bar{E}$ has dynamics \citep{2021_vangoor_AutonomousErrorConstructive}
\begin{align}
    \dot{\bar{E}}
    &= - \bar{E} \Delta - [ \Gamma, \bar{E} ].
    \label{eq:error-dynamics}
\end{align}
In other words, the observer and system are $\bar{E}$-synchronous: the error dynamics depend only on the chosen correction terms $\Delta$ and $\Gamma$, and $\ddt \bar{E} = 0$ if the correction terms are set to zero \citep{2021_vangoor_AutonomousErrorConstructive}.

While the observer architecture \eqref{eq:observer-architecture} allows for $Z$ to be any element of $\SIM_2(3)$,
only some of the degrees of freedom are needed in the proposed design.
Specifically, let $R_Z(0) = I_3$ and $\scale_Z(0) = I_2$, and choose $\Omega_\Gamma = 0$ and $\svel_\Gamma = \svel_D$.
Then $\dot{R}_Z = 0$ and $\dot{\scale}_Z = 0$, and therefore $R_Z \equiv I_3$ and $\scale_Z \equiv I_2$.
It follows that $Z = (I_3, \trans_Z, I_2)$, and $R_Z$ and $\scale_Z$ will not be considered in the sequel.
Under these choices, the error $\bar{E} \in \SE_2(3)$ can be computed as
\begin{align}
    \bar{E}
    &= Z^{-1} X \hat{X}^{-1} Z, \notag \\
    &= \begin{pmatrix} I_3 & \trans_Z \\ 0 & I_2 \end{pmatrix}^{-1}
    \begin{pmatrix} R & \trans \\ 0 & I_2 \end{pmatrix}
    \begin{pmatrix} \hat{R} & \hat{\trans} \\ 0 & I_2 \end{pmatrix}^{-1}
    \begin{pmatrix} I_3 & \trans_Z \\ 0 & I_2 \end{pmatrix}
     , \notag \\
    &= \begin{pmatrix} I_3 & -\trans_Z \\ 0 & I_2 \end{pmatrix}
    \begin{pmatrix} R & \trans \\ 0 & I_2 \end{pmatrix}
    \begin{pmatrix} \hat{R}^\top & - \hat{R}^\top\hat{\trans} \\ 0 & I_2 \end{pmatrix}
    \begin{pmatrix} I_3 & \trans_Z \\ 0 & I_2 \end{pmatrix}
    , \notag \\
    &= \begin{pmatrix} R & \trans - \trans_Z \\ 0 & I_2 \end{pmatrix}
    \begin{pmatrix} \hat{R}^\top & \hat{R}^\top \trans_Z - \hat{R}^\top\hat{\trans} \\ 0 & I_2 \end{pmatrix}
    , \notag \\
    &=\begin{pmatrix} R \hat{R}^\top &\,& R (\hat{R}^\top \trans_Z - \hat{R}^\top\hat{\trans}) + (\trans-\trans_Z) \\ 0 &\,& I_2 \end{pmatrix}
    .
    \label{eq:ebar_expansion}
\end{align}
In summary, if $ \bar{E} = (R_{\bar{E}}, \trans_{\bar{E}}) \in \SE_2(3)$, then
\begin{align}
    R_{\bar{E}} &= R \hat{R}^\top, &
    \trans_{\bar{E}} &= (\trans - R \hat{R}^\top \hat{\trans}) - (I_3 - R \hat{R}^\top) \trans_Z.
    \label{eq:synchronous-error-expansion}
\end{align}

\section{Observer Design}

In the following theorem, the Lie group dynamics \eqref{eq:observer-architecture} are presented in expanded form in equation \eqref{eq:observer_equations}.
However, the formulation of the system and observer as Lie group elements is fundamental to the definition of $\bar{E}$.
This novel error definition is, in turn, core to the design of the correction terms $\Delta$ and $\Gamma$ and the Lyapunov analysis in the proof of the theorem. 

\begin{theorem} \label{thm:observer}
Consider the system $(R, v, p)$ with dynamics \eqref{eq:system_dynamics} and measurement \eqref{eq:system_measurement}.
Define $\hat{R} \in \SO(3)$, $\hat{\trans} = (\hat{v} \; \hat{p}) \in \R^{3 \times 2}$, and $\trans_Z = (v_Z \; p_Z)\in \R^{3\times 2}$, with dynamics
\begin{align}
    \dot{\hat{R}} &= \hat{R} \Omega^\times + \Omega_\Delta^\times \hat{R}, \notag \\
    \dot{\hat{v}} &= \hat{R} a + g + \tvel^1_\Delta + \Omega_\Delta^\times (\hat{v} - v_Z), \notag \\
    \dot{\hat{p}} &= \hat{v} + \tvel^2_\Delta + \Omega_\Delta^\times (\hat{p} - p_Z), \notag \\
    \dot{v}_Z &= g + \tvel_\Gamma^1, \notag \\
    \dot{p}_Z &= v_Z + \tvel_\Gamma^2,
    \label{eq:observer_equations}
\end{align}
where the correction terms $\Omega_\Delta \in \R^3, \tvel_\Delta = (\tvel^1_\Delta \; \tvel^2_\Delta) \in \R^{3 \times 2}, \tvel_\Gamma = (\tvel^1_\Gamma \; \tvel^2_\Gamma) \in \R^{3 \times 2}$ are given by
\begin{align*}
    \Omega_\Delta &= c (\hat{p} - p_Z) \times (p - p_Z), \\
    \tvel_\Delta &= (p - \hat{p}) L, \\
    \tvel_\Gamma &= (p - p_Z) L,
\end{align*}
with $c > 0$ and $L = (l_v \; l_p) \in \R^{1 \times 2}$, $l_p > 0$ and $l_v \in (0, l_p^2 /4)$.

Let $\bar{E} = (R_{\bar{E}}, \trans_{\bar{E}}) \in \SE_2(3)$ be defined as in \eqref{eq:synchronous-error-expansion} and suppose that $R a$ is persistently exciting as in \eqref{eq:persistence_excitation}.
Define
\begin{align} 
\bar{E}_s & := \{(I_3, 0_{3 \times 2})\} \label{eq:Es}\\ 
\bar{E}_u & := \{ (Q, 0_{3\times 2}) \in \SE_2(3) \mid \tr{Q} = -1 \}. \label{eq:Eu}
\end{align}
Then
\begin{enumerate}
    \item The solution $\trans_Z$ is uniformly continuous and bounded for all time and $\bar{E}$ converges to $\bar{E}_s \cup \bar{E}_u$.
    \item The set $\bar{E}_u$ is the set of unstable equilibria.
    \item The singleton set $\bar{E}_s$ is almost-globally asymptotically and locally exponentially stable.
\end{enumerate}
Moreover, if $\bar{E} \in \bar{E}_s$, then $\hat{R} = R$, $\hat{p} = p$, and $\hat{v} = v$.
\end{theorem}

\begin{proof}
The following proof relies on lemmas that are provided in the appendix.

\underline{Proof of item 1}):
Let $C = (0 \; 1)^\top \in \R^2$.
Then the dynamics of $\trans_Z$ may be written as
\begin{align*}
    \dot{\trans} _Z
    &= \trans_Z \svel_D + \begin{pmatrix} g & 0 \end{pmatrix} + (p - p_Z) L, \\
    &= \trans_Z (\svel_D - C L) + \begin{pmatrix} g & 0 \end{pmatrix} + p L, \\
    &= - \trans_Z (C L - \svel_D) + \begin{pmatrix} g + l_v p & l_p p \end{pmatrix}.
\end{align*}
The characteristic equation of the homogeneous ODE is
\begin{align}
    \det(s I_2 + CL - \svel_D)
    &= \det \begin{pmatrix}
        s &  -1 \\ l_v & l_p + s
    \end{pmatrix}, \notag \\
    &= s^2 + l_p s + l_v = 0,
    \label{eq:ode_characteristic_eqn}
\end{align}
with solutions
\begin{align*}
    s = \frac{- l_p \pm \sqrt{l_p^2 - 4 l_v}}{2},
\end{align*}
which are strictly negative for $l_p, l_v > 0$ and $l_v < l_p^2 / 4$.
It follows that $\trans_Z$ is bounded and uniformly continuous.

Expanding \eqref{eq:error-dynamics} yields
\begin{align*}
    \dot{\bar{E}}
    &= \begin{pmatrix}
        R_{\bar{E}} & \trans_{\bar{E}} \\ 0 & I_2
    \end{pmatrix}
    \left(
        \begin{pmatrix}
            0 & \tvel_\Gamma \\ 0 & \svel_D
        \end{pmatrix}
        - \begin{pmatrix}
            \Omega_\Delta^\times & \tvel_\Delta \\ 0 & 0
        \end{pmatrix}
    \right)
    - \begin{pmatrix}
        0 & \tvel_\Gamma \\ 0 & \svel_D
    \end{pmatrix}
    \begin{pmatrix}
        R_{\bar{E}} & \trans_{\bar{E}} \\ 0 & I_2
    \end{pmatrix}, \\
    &= \begin{pmatrix}
        R_{\bar{E}} \Omega_\Delta^\times & R_{\bar{E}}(\tvel_\Gamma - \tvel_\Delta) + \trans_{\bar{E}} \svel_D \\ 0 & \svel_D
    \end{pmatrix}
    - \begin{pmatrix}
        0 & \tvel_\Gamma \\ 0 & \svel_D
    \end{pmatrix}, \\
    &= \begin{pmatrix}
        R_{\bar{E}} \Omega_\Delta^\times & \trans_{\bar{E}} \svel_D + R_{\bar{E}}(\tvel_\Gamma - \tvel_\Delta)- \tvel_\Gamma  \\ 0 & 0
    \end{pmatrix}.
\end{align*}
Hence, the dynamics of $\trans_{\bar{E}}$ are
\begin{align*}
    \dot{\trans}_{\bar{E}}
    &= \trans_{\bar{E}} \svel_D + R_{\bar{E}}(\tvel_\Gamma - \tvel_\Delta)- \tvel_\Gamma, \\
    &= \trans_{\bar{E}} \svel_D + R_{\bar{E}}((p - p_Z) L - (p - \hat{p}) L)- (p - p_Z) L, \\
    &= \trans_{\bar{E}} \svel_D + R_{\bar{E}}(\hat{p} - p_Z) L- (p - p_Z) L, \\
    &= \trans_{\bar{E}} \svel_D + R_{\bar{E}}(\hat{\trans} - \trans_Z) C L- (\trans - \trans_Z) C L, \\
    &= - \trans_{\bar{E}} (CL - \svel_D).
\end{align*}
It follows that $\trans_{\bar{E}}$ converges globally exponentially to zero since the eigenvalues of $-(C L - \svel_D)$ are constant and strictly negative, as shown by the solution of \eqref{eq:ode_characteristic_eqn}.

Recalling \eqref{eq:synchronous-error-expansion}, the dynamics of $R_{\bar{E}}$ can be expanded to
\begin{align*}
    \dot{R}_{\bar{E}}
    &= -R_{\bar{E}} \Omega_\Delta^\times, \\
    &= - c R_{\bar{E}} ((\hat{p} - p_Z) \times (p - p_Z))^\times, \\
    &= - c R_{\bar{E}} (((\hat{\trans} - \trans_Z) C) \times (\trans - \trans_Z) C)^\times, \\
    &= - c R_{\bar{E}} ((R_{\bar{E}}^\top(\trans - \trans_Z - \trans_{\bar{E}}) C) \times (\trans - \trans_Z) C)^\times, \\
    &= - c R_{\bar{E}} ((R_{\bar{E}}^\top(p - p_Z) - R_{\bar{E}}^\top p_{\bar{E}}) \times (p - p_Z))^\times.
\end{align*}

Let $P \in \GL(2)$ diagonalise $C L - \svel_D$; that is, $P (C L - \svel_D) = S P$ where $S = \diag(s_1, s_2)$ is a diagonal matrix of the eigenvalues $s_1 \geq s_2 > 0$ of $C L - \svel_D$.
Consider the Lyapunov function candidate
\begin{align}
    \Lyap(R_{\bar{E}}, \trans_{\bar{E}}) := \trace{I_3 - R_{\bar{E}}} + \frac{\alpha}{2 m_p^2}\vert \trans_{\bar{E}} P \vert^2,
    \label{eq:lyapunov}
\end{align}
where $m_p^2$ is the smallest eigenvalue of $P P^\top$ and $\alpha \geq \frac{c}{2 s_2}$.
The derivative of $\Lyap$ is given by
\begin{align*}
    \dot{\Lyap}
    &= - \tr(\dot{R}_{\bar{E}}) + \alpha m_p^{-2}\left\langle \trans_{\bar{E}} P, \dot{\trans}_{\bar{E}} P \right\rangle, \\
    &= c \tr( R_{\bar{E}} ((R_{\bar{E}}^\top(p - p_Z) - R_{\bar{E}}^\top p_{\bar{E}}) \times (p - p_Z))^\times )
    \\ &\phantom{==}
    - \alpha m_p^{-2} \left\langle \trans_{\bar{E}} P, \trans_{\bar{E}} (C L - \svel_D) P \right\rangle, \\
    &= c \tr( R_{\bar{E}} ((R_{\bar{E}}^\top(p - p_Z)) \times (p - p_Z))^\times )
    \\ &\phantom{==}
    - c \tr( R_{\bar{E}} ((R_{\bar{E}}^\top p_{\bar{E}}) \times (p - p_Z))^\times )
    \\ &\phantom{==}
    - \alpha m_p^{-2} \left\langle \trans_{\bar{E}} P, \trans_{\bar{E}} P S \right\rangle, \\
    &= - \frac{c}{2} \vert (I - R_{\bar{E}}^2) (p - p_Z) \vert^2
    - \alpha m_p^{-2} \left\langle \trans_{\bar{E}} P, \trans_{\bar{E}} P S \right\rangle
    \\ &\phantom{==}
    + c \left\langle
        (I - R_{\bar{E}}^2) (p - p_Z), \;
        p_{\bar{E}}
    \right\rangle, \\
    &\leq - \frac{c}{2} \vert (I - R_{\bar{E}}^2) (p - p_Z) \vert^2
    - \alpha m_p^{-2} s_2 \vert \trans_{\bar{E}} P \vert^2
    \\ &\phantom{==}
    + c \vert (I - R_{\bar{E}}^2)(p - p_Z) \vert
    \vert \trans_{\bar{E}} C \vert, \\
    &\leq - \frac{c}{2} \bigg( \vert (I - R_{\bar{E}}^2) (p - p_Z) \vert^2
    + \vert \trans_{\bar{E}} \vert^2
    \\ &\phantom{==}
    - 2 \vert
    (I - R_{\bar{E}}^2)(p - p_Z) \vert
    \vert \trans_{\bar{E}} \vert \bigg), \\
    &\leq - \frac{c}{2} \left( \vert (I - R_{\bar{E}}^2) (p - p_Z) \vert - \vert \trans_{\bar{E}} \vert \right)^2.
\end{align*}
Thus, the derivative of $\Lyap$ is negative semidefinite, with equality to zero only when $\trans_{\bar{E}} = 0$ and $(I -R_{\bar{E}}^2)(p - p_Z) = 0$.

It is clear that $\dot{\Lyap}$ is uniformly continuous as it is the composition of sums and products of uniformly continuous functions.
Thus Barbalat's lemma \citep[Lemma 4.2/4.3]{1991_slotine_AppliedNonlinearControl} provides that $\Lyap \to \Lyap_{\lim} \geq 0$ and $\dot{\Lyap} \to 0$, with $\Lyap_{\lim} \leq \Lyap(R_{\bar{E}}(0), \trans_{\bar{E}}(0))$ a constant.
Using the fact that $\trans_{\bar{E}} \to 0$ globally exponentially, this means that $\Lyap \to \trace{R_{\bar{E}} - I_3} \to \Lyap_{\lim}$.
Therefore,
\begin{align*}
    \ddt \trace{R_{\bar{E}} - I_3}
    &= - \trace{R_{\bar{E}} \Omega_\Delta^\times}, \\
    &= - \frac{1}{2}\trace{(R_{\bar{E}} - R_{\bar{E}}^\top) \Omega_\Delta^\times} \to 0,
\end{align*}
and, in particular, $\ddt R_{\bar{E}}^2 \to 0$.

Observe that $\ddt (p-p_Z) = Ra - (p-p_Z)(CL - \svel_D)$, and hence $(p-p_Z)$ is persistently exciting by Lemma \ref{lem:persistent_excitation}.
Then $\dot{\Lyap} \to 0$ implies that $R_{\bar{E}}^2(p - p_Z) \to (p - p_Z)$, and hence $R_{\bar{E}}^2 \to I_3$ by \cite[Lemma 5]{2023_vangoor_ConstructiveEquivariantObserver}.
Thus $R_{\bar{E}} \to R_{\bar{E}}^\top$, and hence $R_{\bar{E}} \to I_3$ or $R_{\bar{E}} \to U D U^\top$ for some $U \in \SO(3)$ where $D = \diag(1,-1,-1)$.
In the first case, $\bar{E} \to \bar{E}_s$ asymptotically.
In the second case, one has that $\tr(R_{\bar{E}}) \to \tr(U D U^\top) = \tr(D) = -1$, so $\bar{E} \to \bar{E}_u$ asymptotically.

\underline{Proof of item 2}):
To see that all the elements of $E_u$ are unstable equilibria of the system, let $\bar{E} = (R_{\bar{E}}, 0) \in E_u$ be arbitrary.
It suffices to show that, for any neighbourhood $\calU \subset \SE_2(3)$ of $\bar{E}$, there is an element $(Q, 0) \in \calU$ such that $\Lyap(Q,0) < \Lyap(\bar{E})$.
Since $\tr(R_{\bar{E}}) = -1$, $Q = U D U^\top$, where $D = \diag(1,-1,-1)$ and $U \in \SO(3)$.
Let $\omega = U \eb_1$ and define $Q(s) = R_{\bar{E}} \exp(s \omega^\times)$.
Then, using a second order Taylor expansion provides
\begin{align*}
    \Lyap(Q(s), 0)
    &= \tr(I_3 - Q(s)) + \frac{\alpha}{2 m_p^2} \vert 0 \vert^2, \\
    &= \tr(I_3 - R_{\bar{E}} \exp(s \omega^\times)), \\
    &\approx \tr(I_3 - R_{\bar{E}} (I + s \omega^\times + \frac{s^2}{2}\omega^\times \omega^\times)), \\
    &= \Lyap(R_{\bar{E}}, 0) - s \tr(R_{\bar{E}} \omega^\times) - \frac{s^2}{2} \tr( R_{\bar{E}} \omega^\times \omega^\times), \\
    &= \Lyap(R_{\bar{E}}, 0) - \frac{s^2}{2} \tr( R_{\bar{E}} (\omega \omega^\top - I)), \\
    &= \Lyap(R_{\bar{E}}, 0) - \frac{s^2}{2} (\vert \omega \vert^2 - \tr(R_{\bar{E}})), \\
    &= \Lyap(R_{\bar{E}}, 0) - s^2.
\end{align*}
It follows that, indeed, in any small neighbourhood of $0$, there exists $s$ such that $\Lyap(Q(s), 0) < \Lyap(R_{\bar{E}}, 0)$.
Thus $E_u$ is a set of unstable equilibria, since $\bar{E} \in E_u$ was taken arbitrary.

\underline{Proof of item 3}):
We have already proved that $V_{\bar{E}} \to 0$ globally exponentially and it only remains to show that the attitude part of $\bar{E}$ converges. 
Almost-global asymptotic and locally exponential stability now follows as a consequence of \citep[Theorem 4.3]{2012_trumpf_AnalysisNonLinearAttitude} and the persistence of excitation of $(p - p_Z)$ from Lemma \ref{lem:persistent_excitation}.
Finally, if $\bar{E} = I_5$, then one has that
\begin{align*}
    \hat{X} &= Z^{-1} \bar{E} Z \hat{X} = X \hat{X}^{-1} \hat{X} = X.
\end{align*}
This completes the proof.

\end{proof}

\section{Simulations}

Simulations were conducted to demonstrate the convergence properties of the proposed observer.
The true states of the system were initialised as
\begin{align*}
    R(0) &= I_3,&
    v(0) &= 0_{3 \times 1}~\mbox{m/s},&
    p(0) &= 0_{3 \times 1}~\mbox{m}.
\end{align*}
The input signals were then chosen to be
\begin{align*}
    \Omega(t) &= 1.0 \eb_3~\mbox{rad/s}, &
    a(t) &= 2 \eb_1
    - R^\top (0.75 p + g)~\mbox{m/s$^2$},
\end{align*}
where $g = 9.81 \eb_3$~m/s$^2$.
The observer states were initialised as
\begin{align*}
    \hat{R}(0) &= \exp(0.99\pi \eb_1^\times),\\
    \hat{v}(0) &= \begin{pmatrix}
        0.2 & 0.4 & -1.1
    \end{pmatrix}^\top~\mbox{m/s},\\
    \hat{p}(0) &= \begin{pmatrix}
        3 & -2 & 2
    \end{pmatrix}^\top~\mbox{m},
\end{align*}
and the gains were chosen to be $l_p = 20.0, l_v = 24.0, c = 4.0$.
Both the system and observer equations were simulated for 40~s using Euler integration at 100~Hz.

\begin{figure*}[!htb]
    \centering
    \includegraphics[width=1.0\linewidth]{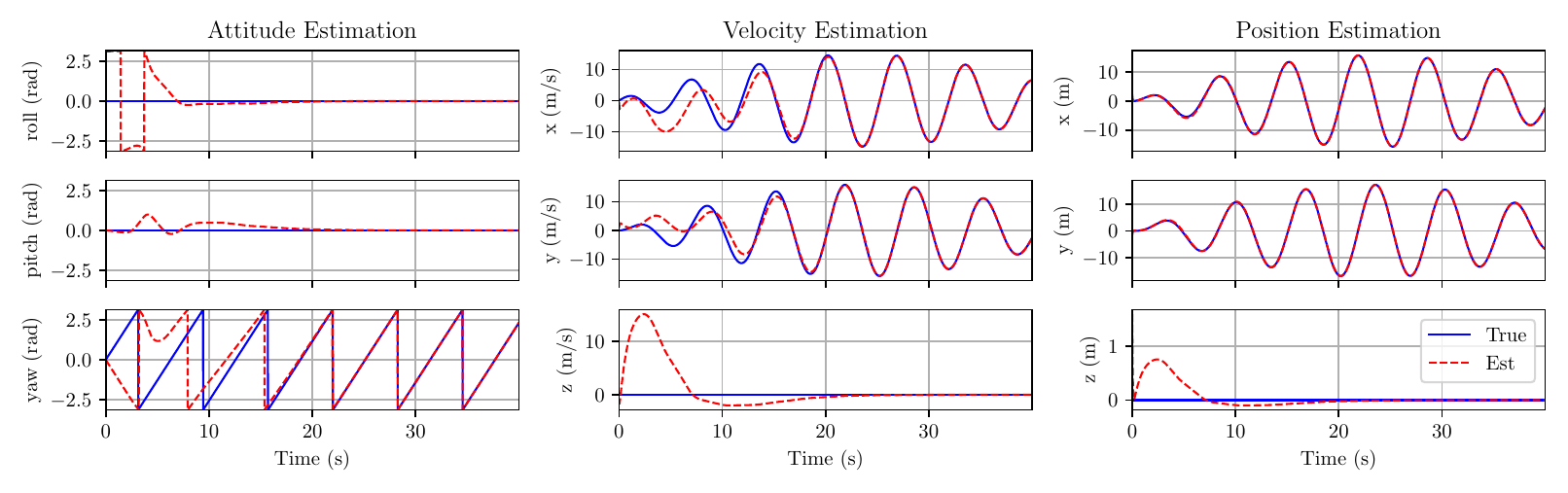}
    \caption{The estimated and true attitude, position, and velocity over time.
    The pitch and roll components of attitude can be seen to converge faster than the yaw.}
    \label{fig:observer_estimates}
\end{figure*}

\begin{figure}[!htb]
    \centering
    \includegraphics[width=0.7\linewidth]{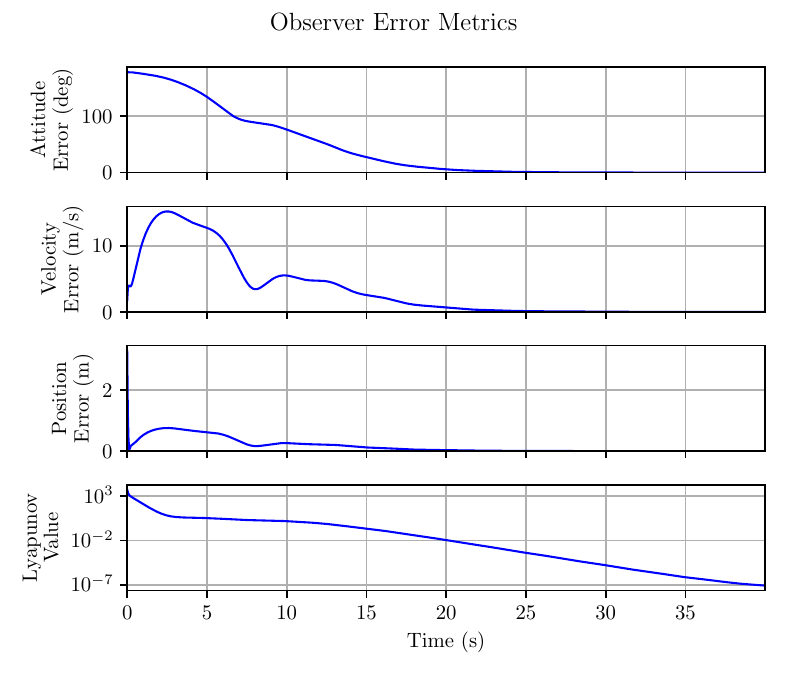}
    \caption{The errors in the observer estimates of attitude, position, and velocity, and the value of the Lyapunov function \eqref{eq:lyapunov} over time.
    The initial decrease is swift and is followed by a slower decrease associated the persistence of excitation condition of Theorem \ref{thm:observer}.}
    \label{fig:observer_errors}
\end{figure}

Figure \ref{fig:observer_estimates} shows the estimated and true attitude, position and velocity over time.
Figure \ref{fig:observer_errors} shows the evolution of observer error metrics over time.
The Lyapunov value is clearly decreasing for all time, verifying the proof of Theorem \ref{thm:observer}.
The position and velocity errors are both seen to converge quickly to zero, although they continue to be disturbed by the direction of attitude that is slower to converge.
The attitude error converges quickly in the first 5 seconds, and then continues to converge more slowly thereafter.
This is explained by the convergence of individual attitude components shown in Figure \ref{fig:observer_estimates}.
The roll and pitch components converge quickly, as they are easily determined through their effect on the measured position.
In contrast, the yaw component converges less quickly as it is only observable through the persistence of excitation of the specific acceleration in the inertial frame $Ra$ required for the convergence result of Theorem \ref{thm:observer}.

\section{Conclusion}

This paper presents an observer design for position-aided INS based on the the authors' recently developed observer architecture for group-affine systems \citep{2021_vangoor_AutonomousErrorConstructive}.
Nonlinear observers for position-aided INS are of particular interest due to their guarantees of stability, which are not available for standard EKF designs.
To the authors' knowledge, the proposed observer is the first solution for position-aided INS with almost-global and local exponential stability properties that are independent of the chosen gains.
Finally, these properties are demonstrated in simulation, where the observer solution is shown to converge even from an extremely poor initial estimate of the system state.

\section*{Appendix}

\begin{lemma}
Let $R \in \SO(3)$ and $x,y \in \R^3$.
Then
\begin{align}
    \tr(R (x \times y)^\times)
    &= -\frac{1}{2} \left\langle (I-R^2) x, \; (I - R^2) R^\top y  \right\rangle.
\end{align}
In particular,
\begin{align}
    \tr(R (x \times (R x))^\times)
    &= - \frac{1}{2}\vert (I-R^2) x \vert^2, \notag \\
    \tr(R ((R^\top x) \times x)^\times)
    &= - \frac{1}{2}\vert (I-R^2) x \vert^2.
\end{align}
\end{lemma}

\begin{proof}
Direct computation yields
\begin{align*}
    \tr(R (x \times y)^\times)
    &= \tr(R ( y x^\top - x y^\top)), \\
    &= \tr(R  y x^\top - R x y^\top), \\
    &= \tr( x^\top R  y - y^\top R x ), \\
    &= \tr( x^\top (R - R^\top) y ), \\
    &= \tr( x^\top (R^2 - I) R^\top y ), \\
    &= -\frac{1}{2}\tr( x^\top (I-R^{2\top})(I - R^2) R^\top y ), \\
    &= -\frac{1}{2} \left\langle (I-R^2) x, \; (I - R^2) R^\top y  \right\rangle.
\end{align*}
\end{proof}

\begin{lemma}
\label{lem:persistent_excitation}
Suppose $a(t) \in \R^3$ is a bounded, uniformly continuous, and persistently exciting signal.
Let $x_1(t), x_2(t) \in \R^3$ satisfy $\dot{x}_1 = -l_1 x_1 + x_2$ and $\dot{x}_2 = -l_2 x_1 + a$ such that $l_1 > 0$ and $l_2 \in (0, l_1^2/4)$.
Then $x_1$ and $x_2$ are also bounded, uniformly continuous, and persistently exciting.
\end{lemma}

\begin{proof}
Define $z = -k x_1 + x_2$ where $k = \frac{1}{2} \sqrt{l_1^2 - 4l_2}$.
Note that $l_1^2 - 4l_2 > 0$ necessarily, so $k > 0$ as well.
Differentiating $z$ yields
\begin{align*}
    \dot{z}
    &= -k \dot{x}_1 + \dot{x}_2, \\
    &= k l_1 x_1 -k x_2 -l_2 x_1 + a, \\
    &= k^2 x_1 -k^2 x_1 + k l_1 x_1 -k x_2 -l_2 x_1 + a, \\
    &= - k z + a -(k^2 - k l_1 +l_2) x_1 , \\
    &= - k z + a.
\end{align*}
Hence $z$ is bounded, uniformly continuous, and persistently exciting by \cite[Lemma A.1]{2023_vangoor_ConstructiveEquivariantObserver}.
Observe that the dynamics of $x_1$ may be written as
\begin{align*}
    \dot{x}_1
    &= -l_1 x_1 + x_2, \\
    &= -l_1 x_1 + k x_1 - k x_1 + x_2, \\
    &= - (l_1 - k) x_1 + z.
\end{align*}
It follows that $x_1$ is bounded, uniformly continuous, and persistently exciting as $k < l_1$ \cite[Lemma A.1]{2023_vangoor_ConstructiveEquivariantObserver}.
\end{proof}

\bibliographystyle{plainnat}
\bibliography{IFAC_2023_INS}

\begin{thebibliography}{17}
\providecommand{\natexlab}[1]{#1}
\providecommand{\url}[1]{\texttt{#1}}
\expandafter\ifx\csname urlstyle\endcsname\relax
  \providecommand{\doi}[1]{doi: #1}\else
  \providecommand{\doi}{doi: \begingroup \urlstyle{rm}\Url}\fi

\bibitem[Barczyk and Lynch(2011)]{2011_barczyk_InvariantExtendedKalman}
Martin Barczyk and Alan~F. Lynch.
\newblock Invariant {{Extended Kalman Filter}} design for a
  magnetometer-plus-{{GPS}} aided inertial navigation system.
\newblock In \emph{2011 50th {{IEEE Conference}} on {{Decision}} and
  {{Control}} and {{European Control Conference}}}, pages 5389--5394, December
  2011.
\newblock \doi{10.1109/CDC.2011.6160733}.

\bibitem[Barrau and Bonnabel(2017)]{2017_barrau_InvariantExtendedKalman}
A.~Barrau and S.~Bonnabel.
\newblock The {{Invariant Extended Kalman Filter}} as a {{Stable Observer}}.
\newblock \emph{IEEE Transactions on Automatic Control}, 62\penalty0
  (4):\penalty0 1797--1812, April 2017.
\newblock ISSN 1558-2523.
\newblock \doi{10.1109/TAC.2016.2594085}.

\bibitem[Berkane and Tayebi(2019)]{2019_berkane_PositionVelocityAttitude}
Soulaimane Berkane and Abdelhamid Tayebi.
\newblock Position, {{Velocity}}, {{Attitude}} and {{Gyro-Bias Estimation}}
  from {{IMU}} and {{Position Information}}.
\newblock In \emph{2019 18th {{European Control Conference}} ({{ECC}})}, pages
  4028--4033, June 2019.
\newblock \doi{10.23919/ECC.2019.8795892}.

\bibitem[Berkane et~al.(2021)Berkane, Tayebi, and {de
  Marco}]{2021_berkane_NonlinearNavigationObserver}
Soulaimane Berkane, Abdelhamid Tayebi, and Simone {de Marco}.
\newblock A nonlinear navigation observer using {{IMU}} and generic position
  information.
\newblock \emph{Automatica}, 127:\penalty0 109513, May 2021.
\newblock ISSN 0005-1098.
\newblock \doi{10.1016/j.automatica.2021.109513}.

\bibitem[Bonnabel(2007)]{2007_bonnabel_LeftinvariantExtendedKalman}
Silvere Bonnabel.
\newblock Left-invariant extended {{Kalman}} filter and attitude estimation.
\newblock In \emph{Procedings of the {{IEEE Conference}} on {{Decision}} and
  {{Control}} ({{CDC}})}, page 6 pages, {New Orleans, LA, USA}, 2007.
\newblock \doi{DOI: 10.1109/CDC.2007.4434662}.

\bibitem[George and Sukkarieh(2005)]{2005_george_TightlyCoupledINS}
Michael George and Salah Sukkarieh.
\newblock Tightly {{Coupled INS}}/{{GPS}} with {{Bias Estimation}} for {{UAV
  Applications}}.
\newblock In \emph{Australasian {{Conference}} on {{Robotics}} and
  {{Automation}}}, page~7, {Sydney, Australia}, December 2005.

\bibitem[Grip et~al.(2012)Grip, Saberi, and
  Johansen]{2012_grip_ObserversInterconnectedNonlinear}
H{\aa}vard~Fj{\ae}r Grip, Ali Saberi, and Tor~A. Johansen.
\newblock Observers for interconnected nonlinear and linear systems.
\newblock \emph{Automatica}, 48\penalty0 (7):\penalty0 1339--1346, July 2012.
\newblock ISSN 0005-1098.
\newblock \doi{10.1016/j.automatica.2012.04.008}.

\bibitem[Grip et~al.(2013)Grip, Fossen, Johansen, and
  Saberi]{2013_grip_NonlinearObserverGNSSaided}
H{\aa}vard~Fj{\ae}r Grip, Thor~I. Fossen, Tor~A. Johansen, and Ali Saberi.
\newblock Nonlinear observer for {{GNSS-aided}} inertial navigation with
  quaternion-based attitude estimation.
\newblock In \emph{2013 {{American Control Conference}}}, pages 272--279, June
  2013.
\newblock \doi{10.1109/ACC.2013.6579849}.

\bibitem[Hansen et~al.(2017)Hansen, Fossen, and
  Arne~Johansen]{2017_hansen_NonlinearObserverDesign}
Jakob~M. Hansen, Thor~I. Fossen, and Tor Arne~Johansen.
\newblock Nonlinear observer design for {{GNSS-aided}} inertial navigation
  systems with time-delayed {{GNSS}} measurements.
\newblock \emph{Control Engineering Practice}, 60:\penalty0 39--50, March 2017.
\newblock ISSN 0967-0661.
\newblock \doi{10.1016/j.conengprac.2016.11.016}.

\bibitem[Mahony et~al.(2008)Mahony, Hamel, and
  Pflimlin]{2008_mahony_NonlinearComplementaryFilters}
R.~Mahony, T.~Hamel, and J.~Pflimlin.
\newblock Nonlinear {{Complementary Filters}} on the {{Special Orthogonal
  Group}}.
\newblock \emph{IEEE Transactions on Automatic Control}, 53\penalty0
  (5):\penalty0 1203--1218, June 2008.
\newblock ISSN 0018-9286.
\newblock \doi{10.1109/TAC.2008.923738}.

\bibitem[Maybeck(1979)]{1979_maybeck_StochasticModelsEstimation}
Peter~S Maybeck.
\newblock \emph{Stochastic Models, Estimation, and Control}, volume~1 of
  \emph{Mathematics in {{Science}} and {{Engineering}}}.
\newblock {Academic Press}, 1979.

\bibitem[Slotine and Li(1991)]{1991_slotine_AppliedNonlinearControl}
J.-J.~E. Slotine and Weiping Li.
\newblock \emph{Applied Nonlinear Control}.
\newblock {Prentice Hall}, {Englewood Cliffs, N.J}, 1991.
\newblock ISBN 978-0-13-040890-7.

\bibitem[Trumpf et~al.(2012)Trumpf, Mahony, Hamel, and
  Lageman]{2012_trumpf_AnalysisNonLinearAttitude}
Jochen Trumpf, Robert Mahony, Tarek Hamel, and Christian Lageman.
\newblock Analysis of {{Non-Linear Attitude Observers}} for {{Time-Varying
  Reference Measurements}}.
\newblock \emph{IEEE Transactions on Automatic Control}, 57\penalty0
  (11):\penalty0 2789--2800, November 2012.
\newblock ISSN 1558-2523.
\newblock \doi{10.1109/TAC.2012.2195809}.

\bibitem[{van Goor} and Mahony(2021)]{2021_vangoor_AutonomousErrorConstructive}
Pieter {van Goor} and Robert Mahony.
\newblock Autonomous {{Error}} and {{Constructive Observer Design}} for {{Group
  Affine Systems}}.
\newblock In \emph{2021 60th {{IEEE Conference}} on {{Decision}} and
  {{Control}} ({{CDC}})}, pages 4730--4737, December 2021.
\newblock \doi{10.1109/CDC45484.2021.9683560}.

\bibitem[{van Goor} et~al.(2023){van Goor}, Hamel, and
  Mahony]{2023_vangoor_ConstructiveEquivariantObserver}
Pieter {van Goor}, Tarek Hamel, and Robert Mahony.
\newblock Constructive {{Equivariant Observer Design}} for {{Inertial
  Velocity-Aided Attitude}}.
\newblock \emph{IFAC-PapersOnLine}, 56\penalty0 (1):\penalty0 349--354, January
  2023.
\newblock ISSN 2405-8963.
\newblock \doi{10.1016/j.ifacol.2023.02.059}.

\bibitem[Vik and Fossen(2001)]{2001_vik_NonlinearObserverGPS}
B.~Vik and T.~Fossen.
\newblock A nonlinear observer for {{GPS}} and {{INS}} integration.
\newblock In \emph{Proceedings of the 40th {{IEEE Conference}} on {{Decision}}
  and {{Control}}}, pages 2956--2961, 2001.

\bibitem[Woodman(2007)]{2007_woodman_IntroductionInertialNavigation}
Oliver~J. Woodman.
\newblock An introduction to inertial navigation.
\newblock Technical Report UCAM-CL-TR-696, {University of Cambridge, Computer
  Laboratory}, 2007.

\end{thebibliography}

\end{document}